\documentclass[letterpaper, 10 pt, conference]{ieeeconf}
\IEEEoverridecommandlockouts        

\usepackage{cite}
\usepackage{amsmath}
\usepackage{amssymb,amsfonts}
\usepackage{algorithmic}
\usepackage{hyperref}
\usepackage{textcomp}
\usepackage{graphicx,color}
\usepackage{mathrsfs}
\usepackage[vlined,ruled]{algorithm2e}
\usepackage{subfigure}
\usepackage{url}
\usepackage{color}
\usepackage{dsfont}
\usepackage{bbm}
\usepackage{booktabs}
\usepackage{array}
\usepackage[table]{xcolor} 
\usepackage{yfonts}

\newtheorem{theorem}{Theorem}[section]

\newtheorem{definition}{Definition}

\newtheorem{proposition}[theorem]{Proposition}
\newtheorem{problem}{Problem}
\newtheorem{remark}{Remark}

\newcommand{\setdef}[2]{\{#1 \; : \; #2\}}

\newcommand{\real}{\mathbb{R}}

\newcommand{\mc}{\mathcal}

\DeclareSymbolFont{bbold}{U}{bbold}{m}{n}
\DeclareSymbolFontAlphabet{\mathbbold}{bbold}

\newcommand\oprocendsymbol{\hbox{$\square$}}
\newcommand\oprocend{\relax\ifmmode\else\unskip\hfill\fi\oprocendsymbol}

\renewcommand{\theenumi}{(\roman{enumi})}

\newcommand*{\QEDA}{\hfill\ensuremath{\blacksquare}}%

\graphicspath{{img/}}

\begin{document}
\title{\bf Data-Driven Attack Detection for Linear
  Systems
  }
  \author{Vishaal Krishnan and Fabio Pasqualetti \thanks{This
    material is based upon work supported in part by
    UCOP-LFR-18-548175 and AFOSR-FA9550-19-1-0235. Vishaal Krishnan
    and Fabio Pasqualetti are with the Department of Mechanical
    Engineering, University of California at Riverside, Riverside, CA
    92521 USA.  E-mail:
    \href{mailto:vishaalk@ucr.edu}{\texttt{vishaalk@ucr.edu}},\href{mailto:fabiopas@engr.ucr.edu}{\texttt{fabiopas@engr.ucr.edu.}}}}
\maketitle

\thispagestyle{empty}

\begin{abstract}
  This paper studies the attack detection problem in a data-driven and
  model-free setting, for deterministic systems with linear and
  time-invariant dynamics. Differently from existing studies that
  leverage knowledge of the system dynamics to derive security bounds
  and monitoring schemes, we focus on the case where the system
  dynamics, as well as the attack strategy and attack location, are
  unknown. We derive fundamental security limitations as a function of
  only the observed data and without estimating the system dynamics
  (in fact, no assumption is made on the identifiability of the
  system). In particular, (i) we derive detection limitations as a
  function of the informativity and length of the observed data, (ii)
  provide a data-driven characterization of undetectable attacks, and
  (iii) construct a data-driven detection monitor. Surprisingly, and
  in accordance with recent studies on data-driven control, our
  results show that model-based and data-driven security techniques
  share the same fundamental limitations, provided that the collected
  data remains sufficiently~informative.
\end{abstract}



\section{Introduction}
\label{sec:intro}
The increasing integration of the
cyber and physical layers in many real-world systems 
has led to the emergence of cyber-physical
systems security as a prominent engineering discipline,
with attack monitoring forming a crucial component.
Its methods differ from traditional information security 
techniques which lack an appropriate
abstraction of the physical layer~\cite{AAC-SA-SS:08}
and are inadequate for the protection of cyber-physical systems,
where the target is often the underlying dynamics of the 
physical system.

Attack monitoring methods for cyber-physical systems 
can be broadly classified into model-based and data-driven 
approaches. 
While the latter relies only on the data, 
generated in the form of measurements from sensors 
deployed on the physical system,
the former additionally assumes knowledge of the model of 
the underlying system.
Clearly, from the perspective of implementation,
model-based monitoring methods~\cite{FP-FD-FB:10y, YC-SK-JMFM:16}
have to first contend with the difficulty 
of obtaining a reliable model for the underlying system.
This is in practice achieved by system identification, for which the
available data must be adequately informative -- a requirement that
is difficult to meet for complex systems.
Moreover, it is unclear if full system identification is even necessary
for attack monitoring.
These considerations have contributed to the increasing popularity in
recent years of data-driven approaches to monitoring.
However, despite the proliferation of data-driven methods for
security, a detailed characterization of their limitations  
is lacking, especially when compared to model-based methods.
This letter fills such a gap for the monitoring 
of linear~systems.

We now provide an overview of the data-driven attack monitoring
problem studied in this letter.  We consider the following
discrete-time linear time-invariant system:
\begin{align}
  \begin{aligned}
    x(k+1) &= A x(k) + B u(k), \\
    y(k) &= C x(k) + D u(k),
  \end{aligned}
           \label{eq:LTI_sys}
\end{align}
with $A \in \real^{n \times n}$, $B \in \real^{n \times m}$,
$C \in \real^{p \times n}$ and $D \in \real^{p \times m}$ the system
matrices, $x: \mathbb{N} \rightarrow \real^n$ the state,
$u: \mathbb{N} \rightarrow \real^m$ the (attack) control input, and
$y: \mathbb{N} \rightarrow \real^p$ the output of the system.
Further, let $u = \left(u_1, \ldots, u_m \right)$, where
$u_j : \mathbb{N} \rightarrow \real$ is the input to actuator
$j \in \lbrace 1, \ldots, m \rbrace$ and
$y = \left( y_1, \ldots, y_p \right)$, where
$y_i : \mathbb{N} \rightarrow \real$ is the output from sensor
$i \in \lbrace 1, \ldots, p \rbrace$.

The attacks are modeled in the form of data injection to 
the system~\eqref{eq:LTI_sys}, which includes a large class
of attacks~\cite{FP-FD-FB:10y}.
Thus, the system~\eqref{eq:LTI_sys} is said to be under attack if 
$u \not\equiv 0$.
The data-driven attack detection problem reads as follows:
\begin{problem}{\bf \emph{ (Data-driven attack
      detection)}}\label{prob:attack_detect}
  Given the output
  $\left \lbrace y(k) \right \rbrace_{k \in \mathbb{N}}$, determine if
  the (attack) input $u \not\equiv 0$.~\oprocend
\end{problem}
\smallskip
An algorithm to solve Problem \ref{prob:attack_detect} is called
\textit{attack monitor}~\cite{FP-FD-FB:10y}. In Problem
\ref{prob:attack_detect}, the matrices $A, B, C, D$ of
system~\eqref{eq:LTI_sys}, the state
$\left \lbrace x(k) \right \rbrace_{k \in \mathbb{N}}$, and the
(attack) input $\left \lbrace u(k) \right \rbrace_{k \in \mathbb{N}}$
are unknown. Further, the Attack Detection
Problem~\ref{prob:attack_detect} is one of binary classification of
time-series of measurements $\lbrace y(k) \rbrace_{k \in \mathbb{N}}$
into the classes
$\lbrace \mathrm{Attack},~ \mathrm{No-Attack} \rbrace$, where
$u \not\equiv 0$ corresponds to the class $\mathrm{Attack}$ and
$u \equiv 0$ to $\mathrm{No-Attack}$. Therefore, data-driven attack
detection is essentially an inverse problem of determining the
structure of the inputs to system~\eqref{eq:LTI_sys} based on the
output measurements. In this letter, we characterize the fundamental
limitations of data-driven attack monitoring when the underlying
system is linear and time-invariant.


\smallskip
\noindent
\textbf{Related work.}  Approaches to data-driven monitoring
fundamentally rely on the fact that the underlying
system~\eqref{eq:LTI_sys} imposes structure on the temporal
characteristics of the output data streams.  The task of a data-driven
attack monitor is then one of detecting changes in
the structure of the output data streams in the absence of a model of
the underlying system that generates them.
Viewed this way, the task of data-driven attack monitoring
is closely related to outlier detection
in sensor data streams~\cite{SP-JS-CF:05, DS-ZG-KHJ-LS:17}.
Recent works have proposed machine learning methods
for data-driven attack monitoring of cyber-physical systems,
which broadly fall within the categories of supervised~\cite{RH-JMB-MAB-TM-UA-SP:14} 
and unsupervised learning~\cite{MK-AS:18, JI-YY-YC-CMP-JS:17,
JG-SA-MT-ZSL:17}.
While these works demonstrate performance of varying degree
in implementation, a theoretical characterization of performance
of these approaches remains elusive.

The performance of data-driven methods relies heavily on the quality
of the available data.  Therefore, an investigation into the
fundamental limitations of data-driven methods must relate notions of
performance to notions of data informativity. In the context of system
identification, persistency of excitation~\cite{JCW-PR-IM-BLMDM:05} is
often assumed as a precondition on the data.  In
\cite{HJVW-JE-HLT-MKC:20}, the authors propose the notion of data
informativity for data-driven control \cite{CDP-PT:19,GB-VK-FP:19},
where they characterize the necessity of persistency of excitation for
data-driven analysis and control and show that for certain problems,
persistency of excitation is not necessary.  In a similar vein, in
this letter we obtain conditions on the information content in the
measurement data in the context of data-driven monitoring.

Finally, in~\cite{FP-FD-FB:10y} the authors characterize undetectable
attacks for model-based attack monitoring. Other works have studied
the problem of false-data injection~\cite{YL-PN-MKR:11, YM-BS:10},
where the objective is to inject inputs that will remain undetected by
an attack monitor. Here, we provide an equivalent characterization of
undetectable attacks in a data-driven setting

\smallskip
\noindent
\textbf{Paper contributions.} 
This letter contributes a characterization of the fundamental limitations of 
data-driven attack monitoring, from a systems-theoretic perspective.
The particular contributions and the outline we follow 
are detailed below:
(i) We first briefly treat the problem of attack detection in the
model-based setting and develop the framework and preliminary results
that we then adapt to the data-driven setting.
(ii) We propose the notion of Hankel information
to characterize the information content in the output
time series. We then obtain a systems-theoretic bound 
on the information content of the output time series 
generated by a given system and the time taken to attain
this bound. 
(iii) We propose a data-driven attack monitoring 
scheme that relies on learning the dynamics 
governing the features in the output data.
(iv) We provide a practical data-driven heuristic 
for handling the output data for use in the attack monitor, 
and characterize the length of the 
time horizon for data collection to reach detection
capability under the heuristic. 
(v) We finally characterize attacks undetectable by
data-driven monitors. 
%
%
%
\section{Data-driven attack detection}
In this section, we address the data-driven attack detection
problem.
The output of system~\eqref{eq:LTI_sys} over a time window
$\lbrace 0, 1, \ldots, N-1 \rbrace$, for any $N \in \mathbb{N}$, can
be expressed as
\begin{align}
	\mathbf{y}_{0:N-1} = \left[ \begin{array}{@{}c|c@{}}
	\mathcal{O}_N & \mathcal{C}_N  \end{array} \right] 
	\left( \begin{matrix}  x(0) \\ \mathbf{u}_{0:N-1}  \end{matrix} \right),
	\label{eq:output_O:N-1}
\end{align}
where $\mathbf{y}_{r:s} = (y(r), y(r+1), \ldots, y(s)) \in \real^{p(s-r+1)}$ and
$\mathbf{u}_{r:s} = (u(r), u(r+1), \ldots, u(s)) \in \real^{m(s-r+1)}$,
for any $r,s \in \mathbb{N}$ such that $r \leq s$.
The matrices $\mathcal{O}_N$ and $\mathcal{C}_N$ read as: \linebreak
\resizebox{.98\linewidth}{!}{
  \begin{minipage}{\linewidth}
\begin{align*}
\mathcal{O}_N = \left[ \begin{matrix}
			C \\ CA \\ CA^2 \\ \vdots \\ CA^{N-1} 
		\end{matrix} \right],~~
\mathcal{C}_N = \left[ \begin{matrix}
			D & 0 & \ldots & 0 \\
			CB & D & \ldots & 0 \\
			CAB & CB & \ldots & 0 \\
			\vdots & \vdots & \ddots & \vdots \\
			CA^{N-2}B & CA^{N-3}B & \ldots & D
						\end{matrix} \right]	.
\end{align*}
\end{minipage}
}
%
%
\linebreak 

\noindent
Notice from~\eqref{eq:output_O:N-1} that the outputs over time-windows
of size $N$ belong to the column space of
$\left[ \begin{array}{@{}c|c@{}} \mathcal{O}_N &
    \mathcal{C}_N \end{array} \right]$.

\subsection{Feature space dynamics and undetectable attacks}
In the ensuing analysis, we consider the nominal system under no
attack (i.e., $u \equiv 0$).  We now factor $\mathcal{O}_N$ by
singular value decomposition \cite{RAH-CRJ:85} and construct
$S^{(\mathrm{m}, \mathrm{nom})}$ from the left singular vectors of
$\mathcal{O}_N$ corresponding to the non-zero singular values,
(where~$\mathrm{m}$ in the superscript denotes that it is obtained
from the model).  Notice that
$\mathrm{Col} \left( S^{(\mathrm{m}, \mathrm{nom})} \right) =
\mathrm{Col} \left( \mathcal{O}_N \right)$, and let
$\mathrm{Col} \left( S^{(\mathrm{m}, \mathrm{nom})} \right)$ be the
\textit{feature space} of the nominal outputs of \eqref{eq:LTI_sys}
over time windows of size~$N$.  Since the columns of
$S^{(\mathrm{m}, \mathrm{nom})}$ are orthonormal, they form a basis of
the nominal feature space.  Then, the nominal output over a time
window $\lbrace k, \ldots, k+N-1 \rbrace$ is
$\mathbf{y}_{k:k+N-1} = \mathcal{O}_N x(k)$, and it can be expressed
in the feature space coordinates as
\begin{align*}
  \mathbf{w}(k) = {S^{(\mathrm{m}, \mathrm{nom})}}^{\top}
  \mathbf{y}_{k:k+N-1} = {S^{(\mathrm{m}, \mathrm{nom})}}^{\top}
  \mathcal{O}_N x(k) .
\end{align*}
We call $\lbrace \mathbf{w}(k) \rbrace_{k \in \mathbb{N}}$ the nominal
\emph{feature vector sequence}. We note that
$\lbrace \mathbf{w}(k) \rbrace_{k \in \mathbb{N}}$ is the
representation of the nominal outputs over windows of size~$N$
of~\eqref{eq:LTI_sys} in the feature space coordinates. It can be
shown, for sufficiently large~$N$, that the nominal feature vector
sequence is generated by the system:

\small
\begin{align}\label{eq:model-based_feature_space_dyn}
  \mathbf{w}(k+1) = \left( {S^{(\mathrm{m}, \mathrm{nom})}}^{\top} \mathcal{O}_N \right) A 
  {\left( {S^{(\mathrm{m}, \mathrm{nom})}}^{\top} \mathcal{O}_N
  \right)}^{\dagger} \mathbf{w}(k), 
\end{align}
\normalsize
where $\mathbf{w}(0) = \mathbf{w}_0 = {S^{(\mathrm{m},
    \mathrm{nom})}}^{\top} \mathcal{O}_N x(0)$.  This forms the
content of the Theorem~\ref{thm:well-posed_feature_dyn} below,
which is proved in
Appendix~\ref{app:proof_prop_feature_space_dyn}. Before stating the
theorem, we recall the notion of observability index of linear
time-invariant systems, which provides a lower bound on the size~$N$
of the time~window. The observability index of the system
\eqref{eq:LTI_sys} is defined as:
  \begin{align*}
    \nu = \min \left \lbrace N \in \mathbb{N} \left| ~\mathrm{Rank}
    \left( \mathcal{O}_{N} \right) = \mathrm{Rank} \left(
    \mathcal{O}_{N + j}\right), ~\forall j \in \mathbb{N}
    \right. \right \rbrace .
  \end{align*}



  \begin{theorem}{\bf \emph{(Feature space
        dynamics)}}\label{thm:well-posed_feature_dyn}
    Let $\nu$ be the observability index of~\eqref{eq:LTI_sys}. For
    any $N \ge \nu$, the feature vector
    sequence~$\lbrace \mathbf{w}(k) \rbrace_{k \in \mathbb{N}}$
    generated by the system~\eqref{eq:LTI_sys} with $u \equiv 0$ is a
    solution to~\eqref{eq:model-based_feature_space_dyn}. \oprocend
\end{theorem}

Theorem~\ref{thm:well-posed_feature_dyn} suggests the design of
an attack detection scheme that is based on verifying if the output
time series $\lbrace y(k) \rbrace_{k \in \mathbb{N}}$, generated by
the system~\eqref{eq:LTI_sys}, can be completely characterized by a
feature vector sequence
$\lbrace \mathbf{w}(k) \rbrace_{k \in \mathbb{N}}$ that is a solution
to~\eqref{eq:model-based_feature_space_dyn}. In particular, if
$\mathbf{y}_{k:k+N-1} = {S^{(\mathrm{m}, \mathrm{nom})}}
\mathbf{w}(k)$ for all $k \in \mathbb{N}$, where
$\lbrace \mathbf{w}(k) \rbrace_{k \in \mathbb{N}}$ is a solution
to~\eqref{eq:model-based_feature_space_dyn}, then
$\lbrace y(k) \rbrace_{k \in \mathbb{N}}$ is classified as
$\mathrm{No-Attack}$. The output sequence is otherwise classified as
$\mathrm{Attack}$. We also note that maximum detection capability is
attained at time $T_\text{safe}^m = \nu$.

\smallskip
\begin{remark}{\bf \emph{(Undetectable attacks)}}
\label{remark:undetectable_model-based} 
  We first note that for time instants $T \leq \nu$, an (attack) input
  $\mathbf{u}_{0:T-1}$ over the time horizon
  $\lbrace 0, \ldots, T-1 \rbrace$ is detectable if and only if
  $\mathcal{C}_{T} \mathbf{u}_{0:T-1} \notin \mathrm{Col} \left(
    \mathcal{O}_{T} \right)$.  Therefore, any (attack) input
  $\mathbf{u}_{0:T-1}$ for which
  $\mathcal{C}_{T} \mathbf{u}_{0:T-1} \in \mathrm{Col} \left(
    \mathcal{O}_{T} \right)$ is undetectable.  However, for attacks
  starting at a time instant~$T \geq \nu +1$ (i.e.,
  $u: \mathbb{N} \rightarrow \real^m$ such that
  $\mathbf{u}_{0:T-1} = 0$), not only must the outputs
  $\mathbf{y}_{k:k+N-1} \in \mathrm{Col} \left( \mathcal{O}_N
  \right)$, they must also, for~$N \geq \nu + 1$, be completely
  characterized by a feature vector sequence
  $\lbrace \mathbf{w}(k) \rbrace_{k \in \mathbb{N}}$ that is a
  solution to~\eqref{eq:model-based_feature_space_dyn}.  Thus, an
  attack starting at $T \geq \nu +1$ is undetectable if and only if
  $\mathbf{u}_{T:T+N-1} \in \mathrm{Ker} \left( \mathcal{C}_N \right)$
  for any $N \in \mathbb{N}$. 
  
  Furthermore, if $u:\mathbb{N} \rightarrow \real^m$ is a finite duration
  attack (i.e., there exists a $d \in \mathbb{N}$ such that
  $\mathbf{u}_{T+d:\infty} = 0$), it can remain undetectable (i.e., 
  $\mathbf{u}_{T:T+N-1} \in \mathrm{Ker} \left( \mathcal{C}_N \right)$
  for any $N \in \mathbb{N}$)
  only if the system~\eqref{eq:LTI_sys} is not left invertible~\cite{MS-JM:69}.
  
  This analysis is compatible with the fundamental limitations derived for
  model-based attack detection in \cite{FP-FD-FB:10y} in terms of the zero dynamics of
  \eqref{eq:LTI_sys}.  \oprocend
\end{remark}

\subsection{Information bound on output time series}
The preceding analysis assumes the knowledge of the system matrix~$A$
and the observability matrix~$\mathcal{O}_N$ of
system~\eqref{eq:LTI_sys}, which are unknown in the data-driven
setting.  We instead have access only to a finite time series of
outputs over a horizon $\lbrace 0, \ldots, T-1 \rbrace$, from which we
can construct the following Hankel matrix with a time window of size
$N \leq T$:
\begin{align}\label{eq:Hankel_matrix}
  Y_{N,T} = \left[ \begin{matrix} y(0) & y(1) & \ldots & y(T-N) \\ y(1) & y(2) & \ldots & y(T-N+1) \\ y(2) & y(3) & &  \\  \vdots & \vdots & \ddots & \vdots \\
      y(N-1) & y(N) & \ldots & y(T-1) \end{matrix} \right].
\end{align}
As a first step, we characterize the information bound on the output
time series for the attack detection problem via an upper bound on the rank of
$Y_{N,T}$. Notice that $\mathrm{Rank} \left( Y_{N,T} \right)$ is a
function of $N, T \in \mathbb{N}$ and the initial state $x(0)$ (since
$\lbrace y(k) \rbrace_{k \in \mathbb{N}}$ is generated by the
underlying system~\eqref{eq:LTI_sys}). We introduce the notion
of Hankel information defined below:
\begin{align}\label{eq:Hankel_info_bound}
  \Gamma\left( \lbrace y(k) \rbrace_{k \in \mathbb{N}} \right) = \sup_{N, T
  \in \mathbb{N}} \setdef{\mathrm{Rank}  \left( Y_{N,T} \right)}{ N \leq T},
\end{align}
as the measure of the information content in the output data
$\lbrace y(k) \rbrace_{k \in \mathbb{N}}$.  In fact, we have
$Y_{k,N,T} = \mathcal{O}_N \left[ \begin{matrix} x(k) & Ax(k) & \ldots
    & A^{T-N} x(k) \end{matrix} \right]$, from which it follows that
$\mathrm{Rank} \left( Y_{N,T} \right) \leq \mathrm{Rank} \left(
  \mathcal{O}_N \right)$.  When
$\mathrm{Rank} \left( Y_{N,T} \right) = \mathrm{Rank} \left(
  \mathcal{O}_N \right)$, we can completely reconstruct the column
space of $\mathcal{O}_N$ from
$\lbrace y(k) \rbrace_{k \in \mathbb{N}}$. This corresponds to the
full information scenario. However, since this is not generally the
case, we provide a bound in Proposition~\ref{prop:information_content} on
the Hankel information of the output data
$\lbrace y(k) \rbrace_{k \in \mathbb{N}}$ generated from an initial
state~$x(0)$. Before
stating the proposition, we first introduce the notion of excitability index
for the system~\ref{eq:LTI_sys}.

\smallskip
\begin{definition}{\bf \emph{(Excitability index)}}
  Let $x \in \real^n$ and $N \in \mathbb{N}$, and let
  $E_N(x) = \left[ \begin{matrix} x & Ax & \ldots & A^{N -
        1}x \end{matrix} \right]$. The excitability index $\mu(x)$
  of~\eqref{eq:LTI_sys} at $x$ is defined as
  \begin{align*}
    \mu(x) = \min \setdef{i \!\!}{\!\!\mathrm{Rank}
    \left( E_i(x) \right) = \mathrm{Rank} \left( E_{i+j}(x) \right),
    \forall j \in \mathbb{N} }.
  \end{align*}
  The excitability index of~\eqref{eq:LTI_sys} is
  $\mu = \max_{x \in \real^n} \mu(x)$. \oprocend
\end{definition}

\smallskip
\begin{proposition}{\bf \emph{(Hankel information bound on output
time series)}}\label{prop:information_content}
  Let $\lbrace y(k) \rbrace_{k \in \mathbb{N}}$ be output of
  \eqref{eq:LTI_sys} from an initial state~$x(0)$. Then, for any
  $N,T \in \mathbb{N}$ with $N \leq T$, we have:
\begin{enumerate}    
 \item {$\mathrm{Rank} \left( Y_{N,T} \right) \leq \min_{\mc V \subseteq \mathrm{Ker}^{\perp} \left(\mathcal{O}_N \right)}
    \setdef{\mathrm{dim}(\mc V)}{x(0) \in \mc V; A \mc V \subseteq \mc V}$},
\item $\Gamma\left( \lbrace y(k) \rbrace_{k \in \mathbb{N}} \right) \leq \min \left \lbrace \nu
    p , \mu \right \rbrace$,
\end{enumerate}
  where $Y_{N,T}$ and~$\Gamma\left( \lbrace y(k) \rbrace_{k \in \mathbb{N}} \right)$
  are as defined in~\eqref{eq:Hankel_matrix} and~\eqref{eq:Hankel_info_bound}.
\end{proposition}
\smallskip
\begin{proof}
  We start by noticing that
  $Y_{N,T} = \mathcal{O}_N \left[ \begin{matrix} x(0) & Ax(0) & \ldots
      & A^{T-N} x(0) \end{matrix} \right]$, from which it follows that
  $\mathrm{Rank} \left( Y_{N,T} \right) \leq \mathrm{Rank} \left(
    \left[ \begin{matrix} x(0) & Ax(0) & \ldots & A^{T-N}
        x(0) \end{matrix} \right] \right)$.  Let $\mc V$ be $A$-invariant,
  and let
  $x(0) \in \mc V \subseteq \mathrm{Ker}^{\perp} \left( \mathcal{O}_N
  \right)$. Then,
  $\mathrm{dim} \left( \left \lbrace x(0), Ax(0), \ldots, A^{T-N} x(0)
    \right \rbrace \right) \leq \mathrm{dim} \left( \mc V \right)$.  It
  then follows that
  $\mathrm{Rank} \left( Y_{N,T} \right) \leq \mathrm{Rank} \left(
    \left[ \begin{matrix} x(0) & Ax(0) & \ldots & A^{T-N}
        x(0) \end{matrix} \right] \right) \leq \mathrm{dim}(\mc V)$. Since
  the above inequality holds for any $A$-invariant
  $\mc V \subseteq \mathrm{Ker}^{\perp} \left( \mathcal{O}_N \right)$, the
  claim follows.

Moreover, we have $\mathrm{Rank} \left(  \mathcal{O}_N \right) \leq \mathrm{Rank} \left(  \mathcal{O}_{\nu} \right)$
and $\mathrm{Rank} \left( \left[ \begin{matrix}	x(0) & Ax(0) & \ldots & A^{T-N}	x(0) \end{matrix} \right] \right) \leq 
\mathrm{Rank} \left( \left[ \begin{matrix}	x(0) & Ax(0) & \ldots & A^{\mu - 1}	x(0) \end{matrix} \right] \right)$
where $\nu$ and $\mu$ are the observability and excitability indices respectively.
Moreover, we have $\mathrm{Rank} \left( Y_{N,T} \right) =
\mathrm{dim} \left( \mathrm{Ker}^{\perp}\left( \mathcal{O}_N \right) \cap \mathrm{Col} \left( 
\left[ \begin{matrix}	x(0) & Ax(0) & \ldots & A^{T-N}	x(0) \end{matrix} \right] \right) \right)$.
Since $\mathrm{dim} \left( \mathrm{Ker}^{\perp}\left( \mathcal{O}_N \right) \right) =  
\mathrm{dim} \left( \mathrm{Col} \left( \mathcal{O}_N \right) \right)\leq \min \lbrace \nu p, n \rbrace$,
$\mathrm{dim} \left( \mathrm{Col} \left( \left[ \begin{matrix}	x(0) & Ax(0) & \ldots & A^{T-N}	x(0) \end{matrix} \right] \right) \right) 
\leq \mu$ for any $N, T \in \mathbb{N}$ and $x(0) \in \real^n$, and 
we have $\mu \leq n$, 
we get $\sup \setdef{\mathrm{Rank}  \left( Y_{N,T} \right) \!\!}{ \!\! N, T
    \in \mathbb{N} \text{ and } N \leq T} 
\leq \min \left \lbrace \nu p , \mu \right \rbrace$.
\end{proof}

Proposition~\ref{prop:information_content} implies that, if the initial
state~$x(0)$ belongs to an $A$-invariant subspace of
$\mathrm{Ker}^{\perp} \left( \mathcal{O}_N \right)$ of lower
dimension, then the column space of $\mathcal{O}_N$ cannot be
completely characterized from the Hankel matrix~$Y_{N,T}$, for any
$T \in \mathbb{N}$.
Proposition~\ref{prop:information_content} therefore characterizes an
information bound on $\lbrace y(k) \rbrace_{k \in \mathbb{N}}$. We
now characterize the length of the shortest time horizon over which
$\lbrace y(k) \rbrace_{k \in \mathbb{N}}$ attains 
the Hankel information bound.

\smallskip
\begin{theorem}{\bf \emph{(Minimum horizon length)}}\label{thm:minimum_hor_length}
  Let $\lbrace y(k) \rbrace_{k \in \mathbb{N}}$ be the outputs of
  \eqref{eq:LTI_sys} from the initial state~$x(0)$. Then, for any
  $x(0) \in \real^n$ and $T_0 \geq \nu + \mu - 1$, there exists
  $N_0 \in \mathbb{N}$ with $N_0 \leq T_0$ such that:
\begin{align*}
  \mathrm{Rank} \left( Y_{N_0,T_0} \right) = 
   \Gamma\left( \lbrace y(k) \rbrace_{k \in \mathbb{N}} \right),
\end{align*} 
where $Y_{N,T}$ and~$\Gamma\left( \lbrace y(k) \rbrace_{k \in \mathbb{N}} \right)$
  are as defined in~\eqref{eq:Hankel_matrix} and~\eqref{eq:Hankel_info_bound}.
\end{theorem}
\begin{proof}
For the Hankel matrix~$Y_{N,T}$,
we have $r = Np$ rows and let~$c$ be the number of columns.
Then by construction, we have 
$T = N + c - 1$.
Thus, we have
$Y_{N,T} = \mathcal{O}_N  \left[ \begin{matrix}	x(0) & Ax(0) & \ldots & A^{c-1}	x(0) \end{matrix} \right]$,
and 
$\mathrm{Rank} \left( Y_{N,T} \right) =
\mathrm{dim} \left( \mathrm{Ker}^{\perp}\left( \mathcal{O}_N \right) \cap \mathrm{Col} \left( 
\left[ \begin{matrix}	x(0) & Ax(0) & \ldots & A^{c-1}	x(0) \end{matrix} \right] \right) \right)$.
We also have $\mathrm{Ker}^{\perp}\left( \mathcal{O}_N \right) \subseteq 
\mathrm{Ker}^{\perp}\left( \mathcal{O}_{\nu} \right)$
and $\mathrm{Col} \left( \left[ \begin{matrix}	x(0) & Ax(0) & \ldots & A^{c-1}	x(0) \end{matrix} \right] \right)
\subseteq \mathrm{Col} \left( \left[ \begin{matrix}	x(0) & Ax(0) & \ldots & A^{\mu-1}	x(0) \end{matrix} \right] \right)$
for all $N, c \in \mathbb{N}$.
Therefore, we get $\mathrm{Ker}^{\perp}\left( \mathcal{O}_N \right) \cap \mathrm{Col} \left( 
\left[ \begin{matrix}	x(0) & Ax(0) & \ldots & A^{c-1}	x(0) \end{matrix} \right] \right) 
\subseteq \mathrm{Ker}^{\perp}\left( \mathcal{O}_{\nu} \right) \cap \mathrm{Col} \left( 
\left[ \begin{matrix}	x(0) & Ax(0) & \ldots & A^{\mu-1}	x(0) \end{matrix} \right] \right) $
Thus, for $T_0 \geq \nu + \mu - 1$,
we can choose $N_0 = \nu$ and $c = \mu$
such that $Y_{N_{0},T_0}$ attains maximum rank.
\end{proof}

Theorem~\ref{thm:minimum_hor_length} states that for any initial
state~$x(0)$, we attain the Hankel information (upper) bound in
Proposition~\ref{prop:information_content} for the time series of outputs
$\lbrace y(k) \rbrace_{k \in \mathbb{N}}$ generated by the
system~\eqref{eq:LTI_sys}, within the finite time horizon
$\lbrace 0, \ldots, \nu + \mu -1 \rbrace$.


\subsection{Data-driven attack detection monitor}
\label{sec:data-driven_attack_monitor}
For the nominal system under no attack (i.e., $u \equiv 0$), we now
seek to obtain a data-driven expression for the feature space
dynamics, as in~\eqref{eq:model-based_feature_space_dyn}.
We factor the Hankel matrix $Y_{N,T} \in \real^{pN \times (T-N+1)}$ by
singular value decomposition to obtain a matrix
$S_{N,T} \in \real^{pN \times q}$ whose column vectors are the
(orthonormal) left singular vectors of $Y_{N,T}$ corresponding to the
non-zero singular values. The columns of $S_{N,T}$ can be interpreted
as the \emph{features in the output data}
$\lbrace y(k) \rbrace_{k=0}^{T-1}$ with a time window of size $N$.
In the ensuing analysis, we fix $N$, $T$ and suppress them from the
notation, hereby denoting $Y_{N,T}$ and $S_{N,T}$ simply by $Y$ and
$S$, respectively.

We now define feature vectors $\mathbf{w}(k) = S^{\top} \mathbf{y}_{k:k+N-1}$
by projecting the outputs $\mathbf{y}_{k:k+N-1}$, for $k \in \lbrace 0, \ldots, T-N \rbrace$,
onto the feature space, and construct a matrix $W$ as follows:
\begin{align*}
	W = \left[ \begin{matrix}  \vert & \vert &  & \vert \\ \mathbf{w}(0) & \mathbf{w}(1) & \ldots & \mathbf{w}(T-N) 
	\\ \vert & \vert &  & \vert  \end{matrix}	 \right] .
\end{align*}
We then similarly construct a matrix $\vec{W}$ from output data over a time horizon $\lbrace 1, \ldots, T \rbrace$,
i.e., $\lbrace y(k) \rbrace_{k=1}^T$, to get:
\begin{align*}
	\vec{W} = \left[ \begin{matrix}  \vert & \vert &  & \vert \\ \mathbf{w}(1) & \mathbf{w}(2) & \ldots & \mathbf{w}(T-N+1) 
	\\ \vert & \vert &  & \vert  \end{matrix}	 \right] .
\end{align*}
We show in Theorem~\ref{thm:well-posed_data-driven_feature_dyn} that,
for~$T$ greater than the minimum horizon length in
Theorem~\ref{thm:minimum_hor_length}, the nominal feature vector
sequence is obtained as a solution to
$\mathbf{w}(k+1) = M^* \mathbf{w}(k)$ for all $k \in \mathbb{N}$, where
$M^* \in \real^{q \times q}$ is a solution to the following minimization
problem:
%
\begin{align}
	M^* = \arg \min_{M \in \real^{q \times q}} \| \vec{W} - M W \|_{F},
	\label{eq:feature_space_dyn_minimization}
\end{align} 
where $\| \cdot \|_F$ denotes the Frobenius norm.  In the data-driven
analysis literature, this approach goes by the name of Dynamic Mode
Decomposition~\cite{JHT-CWR-DML-SLB-JNK:14}, and is closely related to
the Eigensystem Realization Algorithm~\cite{JNJ-RSP:85} in the domain
of system identification.
Theorem~\ref{thm:well-posed_data-driven_feature_dyn} characterizes the
conditions under which the feature vector sequence offers a complete
characterization of the nominal output time series.

\smallskip
\begin{theorem}{\bf \emph{(Data-driven feature space
      dynamics)}}\label{thm:well-posed_data-driven_feature_dyn}
  For $N \geq \nu$ and $T \geq N + \mu -1$, the feature vector
  sequence $\lbrace \mathbf{w}(k) \rbrace_{k \in \mathbb{N}}$
  generated by the (nominal) system~\eqref{eq:LTI_sys} with
  $u \equiv 0$ is a solution to $\mathbf{w}(k+1) = M^* \mathbf{w}(k)$
  for all $k \in \mathbb{N}$, where $M^*$ is given
  by~\eqref{eq:feature_space_dyn_minimization}.  Moreover, the global
  minimizer in~\eqref{eq:feature_space_dyn_minimization} is
  $M^* = \vec{W} W^{\dagger}$, and the minimum value is $0$.
\end{theorem}

The proof of Theorem~\ref{thm:well-posed_data-driven_feature_dyn} is
provided in
Appendix~\ref{app:proof_thm_well-posed_data-driven_feature_dyn}.
It follows from Theorem~\ref{thm:well-posed_data-driven_feature_dyn}
that, if the system is not under attack over the time horizon
$\lbrace 0, 1, \ldots, T \rbrace$ with $T \geq \nu + \mu - 1$, the
feature vector sequence $\lbrace \mathbf{w}(k) \rbrace_{k=0}^T$ is a
solution to $\mathbf{w}(k+1) = M^* \mathbf{w}(k)$, where $M^*$ can be
computed from the (nominal) output sequence
$\lbrace y(k) \rbrace_{k=0}^T$ via the
minimization~\eqref{eq:feature_space_dyn_minimization}.  This enables
the prediction of the outputs starting at time $T+1$, as
$\hat{\mathbf{y}}_{k:k+N-1} = S \mathbf{w}(k)$ for
$k \in \lbrace T-N+1, \ldots \rbrace$. An attack detection scheme can
therefore be defined based on the prediction error, where the case
$\mathbf{y}_{k:k+N-1} \neq \hat{\mathbf{y}}_{k:k+N-1}$ is
classified as $\mathrm{Attack}$.
Moreover, the maximum data-driven attack detection capability is
attained at time instant $T_{\text{safe}} = \nu + \mu$. In comparison, in the
model-based setting attack detection capability is attained at time
$T_{\text{safe}}^m = \nu$, where~$\nu$ is the observability index,
as established in Theorem~\ref{thm:well-posed_feature_dyn}
and discussed in Remark~\ref{remark:undetectable_model-based}.
We refer the reader to~\cite{FP-FD-FB:10y} for detailed accounts on 
model-based attack monitor design.
%


Although attack detection capability can be attained at
$T_{\text{safe}} = \nu + \mu$, it requires the construction of a
Hankel matrix of appropriate dimensions (with $\nu p$ rows and $\mu$
columns). However, since~$\nu$ and~$\mu$ are unknown in the
data-driven setting, achieving full detection capability at
$T_{\text{safe}} = \nu + \mu$ requires an algorithm to compute~$N$,
which is likely to increase the time complexity of the monitor.  We
therefore provide a practical data-driven heuristic in
Remark~\ref{remark:heuristic_N} for the choice of~$N$, and an estimate
in Theorem~\ref{thm:safe_time_heuristic} for the
time~$\bar{T}_\text{safe}$ at which full attack detection capability
is attained for the heuristic.

\begin{remark}{\bf \emph{(Data-driven choice of
      $N$)}}\label{remark:heuristic_N}
  Given a finite time series of outputs over a time horizon
  $\lbrace 0, \ldots, T-1 \rbrace$, and a window of size~$N$, the
  Hankel matrix satisfies $Y_{N,T} \in \real^{pN \times (T-N+1)}$. In
  our heuristic algorithm, we select $N$ to satisfy $T-N+1 \geq pN$,
  to maintain at least as many columns as there are rows. This leads
  to the choice $N = \left \lfloor \frac{T+1}{p+1} \right
  \rfloor$. \oprocend
\end{remark}

\smallskip
\begin{theorem}{\bf \emph{(Safe time horizon
      length)}}\label{thm:safe_time_heuristic}
  Let $\lbrace y(k) \rbrace_{k \in \mathbb{N}}$ be the output of
  \eqref{eq:LTI_sys}, let
  $N(T) = \left \lfloor \frac{T+1}{p+1} \right \rfloor$, and let
  $\bar{T}_\text{safe} \geq \max \left \lbrace \nu (p+1) - 1, \mu
    \left(\frac{p+1}{p}\right) -1 \right \rbrace$. Then,
  \begin{align*}
    \mathrm{Rank} \left( Y_{N(\bar{T}_{\text{safe}}),\bar{T}_{\text{safe}}} \right)
    = \Gamma\left( \lbrace y(k) \rbrace_{k \in \mathbb{N}} \right),
  \end{align*}
  where $Y_{N,T}$ and~$\Gamma\left( \lbrace y(k) \rbrace_{k \in \mathbb{N}} \right)$
  are as defined in~\eqref{eq:Hankel_matrix} and~\eqref{eq:Hankel_info_bound}.
\end{theorem}
\begin{proof}
  Following the arguments in the proof of
  Theorem~\ref{thm:minimum_hor_length}, the observability
  index~$\nu$ yields a necessary and sufficient lower bound on~$N$,
  and we get a lower bound on $T$ by allowing
  $\left \lfloor \frac{T+1}{p+1} \right \rfloor \geq \nu$.  Moreover,
  the excitability index yields a necessary and sufficient lower
  bound on the number of columns as $T-N+1 \geq \mu$.  Therefore, for
  a safe time horizon length $\bar{T}_{\text{safe}}$ satisfying
  $\bar{T}_{\text{safe}} \geq \max \left \lbrace \nu (p+1) - 1, \mu
    \left(\frac{p+1}{p}\right) -1 \right \rbrace$, we get that
  $\mathrm{Rank} \left( Y_{N(\bar{T}_{\text{safe}}),\bar{T}_{\text{safe}}} \right)$
  attains the maximum value.
\end{proof}

We next characterize undetectable attacks for the data-driven attack
monitor. Notice that any attack over the time horizon
$\lbrace 0, \ldots, \nu + \mu \rbrace$ is undetectable, since attack
detection capability is only attained at $T_\text{safe}= \nu +
\mu$. We therefore restrict our attention to attacks starting at
$T \geq T_\text{safe} + 1$.

\smallskip
\begin{theorem}{\bf \emph{(Data-driven undetectable
      attacks)}}\label{thm:undetectable_data-driven}
  The (attack) input $u : \mathbb{N} \rightarrow \real^m$ to the
  system~\eqref{eq:LTI_sys}, with $\mathbf{u}_{0:T-1} = 0$ for
  $T \geq T_{\text{safe}}+1 = \nu + \mu + 1$, is undetectable if and only if
  $\mathbf{u}_{T:T+N-1} \in \mathrm{Ker} \left( \mathcal{C}_N \right)$
  for all $N \in \mathbb{N}$.
\end{theorem}
\begin{proof}
We have $\mathbf{w}(T-N) = S^{\top} \mathbf{y}_{T-N:T-1}$
and from the feature space dynamics, we get $\mathbf{w}(T-N+1)
= M^* \mathbf{w}(T-N)$, where~$M^*$ is obtained from~\eqref{eq:feature_space_dyn_minimization}
over the time horizon $\lbrace 0, \ldots, T_{\text{safe}}\rbrace$
and the predicted output over $\lbrace T-N+1, \ldots, T \rbrace$
is given by $\hat{\mathbf{y}}_{T-N+1:T} = S \mathbf{w}(T-N+1) = SM^*S^{\top} \mathbf{y}_{T-N:T-1}$.
It follows from~\eqref{eq:output_O:N-1} and Theorem~\ref{thm:well-posed_data-driven_feature_dyn}
that the output prediction error $\mathbf{y}_{T-N+1:T} - \hat{\mathbf{y}}_{T-N+1:T} = \mathcal{C}_N \mathbf{u}_{T-N+1:T}$,
where $\mathbf{u}_{T-N+1:T} = \left(0, \ldots, 0, u(T) \right)$.
For an undetectable attack, the prediction error  
$\mathbf{y}_{T-N+1:T} - \hat{\mathbf{y}}_{T-N+1:T} = \mathcal{C}_N \mathbf{u}_{T-N+1:T} = 0$,
which implies that $\mathbf{u}_{T-N+1:T} \in \mathrm{Ker} \left( \mathcal{C}_N \right)$,
i.e., $u(T) \in \mathrm{Ker} \left( D \right)$.
By induction, we get that $\mathbf{u}_{T:T+N-1} \in \mathrm{Ker} \left( \mathcal{C}_N \right)$
for any $N \in \mathbb{N}$.
Conversely, if $\mathbf{u}_{T:T+N-1} \in \mathrm{Ker} \left( \mathcal{C}_N \right)$
for any $N \in \mathbb{N}$, we see that the prediction error
vanishes, which implies that the attack is undetectable,
thereby proving the claim.
\end{proof}
%
%
\section{Numerical experiments}
We now present results from numerical experiments 
validating the key theoretical results presented in this letter
and comparing model-based and data-driven attack monitoring.  
We considered a linear time-invariant 
system~\eqref{eq:LTI_sys} of state-space
dimension $n = 50$ defined below: 
\begin{align*}
	A = \left[ \begin{matrix}
\mathbf{0}_{(n-1) \times 1} & \vline &  I_{n-1} \\ \hline 
-1 & \vline & - \mathbf{1}_{1 \times (n-1)}
\end{matrix} \right],
\end{align*} 
with $m = 5$ actuators and $p = 10$ sensors.
%
%
The columns of~$B$ and the
rows of~$C$ were distinctly chosen at random from
$\lbrace e_i \rbrace_{i=1}^{n}$, the standard basis for~$\real^{n}$.
The matrix~$D$ was chosen to be the zero matrix.

We determined the values of the observability and
excitability indices numerically to
be~$\nu = 15$ and~$\mu = 50$, respectively.
With the data-driven heuristic
from Remark~\ref{remark:heuristic_N} 
we plot in Figure~\ref{fig:rank_Y}
the rank of the Hankel matrix~$Y_{N(T),T}$,
$\mathrm{Rank}\left( Y_{N(T),T} \right)$
as a function of $T$.
\begin{figure}[!h]
	\begin{center}
             \includegraphics[width=0.5\textwidth]{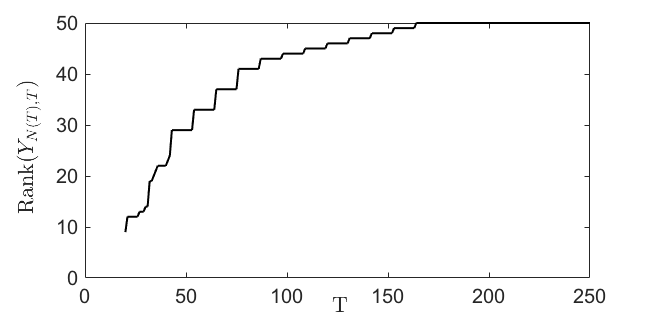}
	\end{center}
        \caption{The figure shows the plot of the rank of the Hankel
          matrix~$Y_{N(T),T}$ vs. time~$T$ with~$N$ chosen according
          to the data-driven heuristic from
          Remark~\ref{remark:heuristic_N}.  We observe that
          $\mathrm{Rank}(Y_{N(T),T})$ is a monotonically increasing
          function of $T$, and a maximum rank of $50$ is attained at
          $\bar{T}_{\text{safe}}=164$. With $p=10$ and $\nu=15$
          obtained numerically, we see that
          $\nu(p+1)-1 = 164 = \bar{T}_{\text{safe}}$, which validates
          Theorem~\ref{thm:safe_time_heuristic}. }
	\label{fig:rank_Y}
\end{figure}
We then designed model-based and data-driven attack detection monitors
based on feature space dynamics~\eqref{eq:model-based_feature_space_dyn}
and as outlined in Section~\ref{sec:data-driven_attack_monitor}, respectively.
%
%
We injected the system with an attack input at the actuator~$j=4$
at $T = 249 > \bar{T}_{\text{safe}}$, such that $u_4(249) = 1$,
and tracked the prediction error for the model-based and
data-driven attack detection monitors starting from 
a random initial state,
as shown in Figure~\ref{fig:attack_detect}.
\begin{figure}[!h]
	\begin{center}
        \hspace*{-0.2cm}     \includegraphics[width=0.52\textwidth]{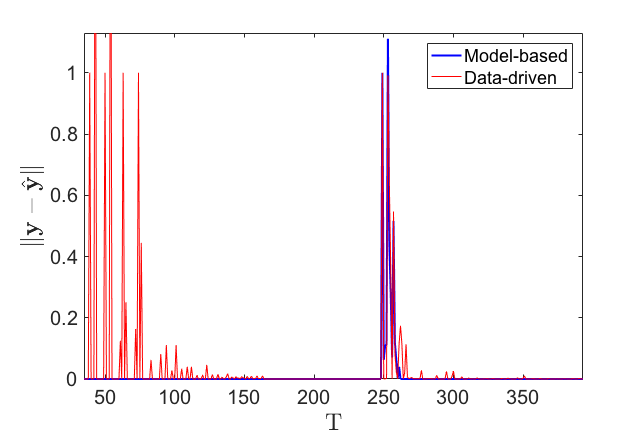}
	\end{center}
   \caption{The figure shows the responses of the model-based and
   data-driven attack monitors for an attack input $u_4(249) = 1$.
  Attack detection capability for the model-based monitor 
	is achieved at time $T = 15$ and for data-driven monitor
    at time $\bar{T}_{\text{safe}}=164$ (consistent with the result in Figure~\ref{fig:rank_Y}) 
   as seen from the convergence to zero of the prediction error,
   implicitly validating Theorems~\ref{thm:well-posed_feature_dyn} and
   \ref{thm:well-posed_data-driven_feature_dyn}.
   The attack is detected as an error in the output prediction by the
   monitors at $T = 249$.
   We also observe that the data-driven monitor recovers from the
   attack at around $T = 270$ when the prediction error is below~$0.03$ units
   ($< 3\%$ of attack input magnitude) and gradually decays to zero.}
	\label{fig:attack_detect}
\end{figure}
\section{Conclusion and future work}
In this letter, we characterized the fundamental limitations on
data-driven monitoring of linear time-invariant systems from a
systems-theoretic perspective.
In particular: (i) we characterized an information bound on the output
time series and the minimum time horizon length for measurement data
collection to achieve detection capability; (ii) we provided a
heuristic for the choice of dimensions of the Hankel matrix employed
in data-driven attack detection; and (iii) we obtained a
characterization of undetectable attacks. Surprisingly, our results
show that model-based and data-driven detection strategies share the
same limitations, thus relaxing the stringent assumption of the
knowledge of the system dynamics that is typically made in existing
security works.

Future work includes the limitations of data-driven detection in
stochastic, time-varying and nonlinear systems.

%
\bibliographystyle{unsrt}
\bibliography{alias,Main,FP,New}
\begin{appendices}
\section{Proof of Theorem~\ref{thm:well-posed_feature_dyn}}
\label{app:proof_prop_feature_space_dyn}
We have $\mathbf{w}(k+1) = {S^{(\mathrm{m}, \mathrm{nom})}}^{\top} \mathcal{O}_N x(k+1) 
= {S^{(\mathrm{m}, \mathrm{nom})}}^{\top} \mathcal{O}_N A x(k)$.
Suppose $x(k) \notin \mathrm{Ker}\left( \mathcal{O}_N \right)$,
we get that $0 \neq \mathcal{O}_N x(k) \in \mathrm{Col}\left( \mathcal{O}_N \right) 
= \mathrm{Col} \left( S^{(\mathrm{m}, \mathrm{nom})} \right)$, and we can express
$\mathcal{O}_N x(k)$ as a linear combination of the columns of $S^{(\mathrm{m}, \mathrm{nom})}$.
It then follows that ${S^{(\mathrm{m}, \mathrm{nom})}}^{\top} \mathcal{O}_N x(k) \neq 0$.
Moreover, if $x(k) \in \mathrm{Ker} \left( \mathcal{O}_N \right)$, then we have
${S^{(\mathrm{m}, \mathrm{nom})}}^{\top} \mathcal{O}_N x(k) \neq 0$. Therefore, 
we get that $\mathrm{Ker} \left( {S^{(\mathrm{m}, \mathrm{nom})}}^{\top} \mathcal{O}_N \right)
= \mathrm{Ker} \left( \mathcal{O}_N \right)$.

We note that $\mathbf{w}(k) \in \mathrm{Col} \left( {S^{(\mathrm{m}, \mathrm{nom})}}^{\top} \mathcal{O}_N \right)
\cong \mathrm{Ker}^{\perp} \left( {S^{(\mathrm{m}, \mathrm{nom})}}^{\top} \mathcal{O}_N \right)$ for all $k \in \mathbb{N}$.
We now express $x(k) \in \real^n =  \mathrm{Ker} \left( {S^{(\mathrm{m}, \mathrm{nom})}}^{\top} \mathcal{O}_N \right) 
\oplus \mathrm{Ker}^{\perp} \left( {S^{(\mathrm{m}, \mathrm{nom})}}^{\top} \mathcal{O}_N \right)$ as 
$x(k) = \alpha(k) + \beta(k)$, where $\alpha(k) \in \mathrm{Ker} \left( {S^{(\mathrm{m}, \mathrm{nom})}}^{\top} \mathcal{O}_N \right)$
and $\beta(k) \in \mathrm{Ker}^{\perp} \left( {S^{(\mathrm{m}, \mathrm{nom})}}^{\top} \mathcal{O}_N \right)$.
We have $\mathbf{w}(k) = {S^{(\mathrm{m}, \mathrm{nom})}}^{\top} \mathcal{O}_N x(k) = {S^{(\mathrm{m}, \mathrm{nom})}}^{\top} \mathcal{O}_N \beta(k)$.
Since ${S^{(\mathrm{m}, \mathrm{nom})}}^{\top} \mathcal{O}_N$ is a bijection from 
$\mathrm{Ker}^{\perp} \left( {S^{(\mathrm{m}, \mathrm{nom})}}^{\top} \mathcal{O}_N \right)$ to 
$\mathrm{Col} \left( {S^{(\mathrm{m}, \mathrm{nom})}}^{\top} \mathcal{O}_N \right)$, with ${\left( {S^{(\mathrm{m}, \mathrm{nom})}}^{\top} \mathcal{O}_N \right)}^{\dagger}$
as its inverse, we get that $\beta(k) = {\left( {S^{(\mathrm{m}, \mathrm{nom})}}^{\top} \mathcal{O}_N \right)}^{\dagger} \mathbf{w}(k)$.

Now, we have $\mathbf{w}(k+1) = {S^{(\mathrm{m}, \mathrm{nom})}}^{\top} \mathcal{O}_N x(k+1) 
= {S^{(\mathrm{m}, \mathrm{nom})}}^{\top} \mathcal{O}_N A x(k) = {S^{(\mathrm{m}, \mathrm{nom})}}^{\top} \mathcal{O}_N A \left( \alpha(k) + \beta(k) \right)
= {S^{(\mathrm{m}, \mathrm{nom})}}^{\top} \mathcal{O}_N A \alpha(k) + {S^{(\mathrm{m}, \mathrm{nom})}}^{\top} \mathcal{O}_N A \beta(k) 
= {S^{(\mathrm{m}, \mathrm{nom})}}^{\top} \mathcal{O}_N A \alpha(k) + {S^{(\mathrm{m}, \mathrm{nom})}}^{\top} \mathcal{O}_N A 
{\left( {S^{(\mathrm{m}, \mathrm{nom})}}^{\top} \mathcal{O}_N \right)}^{\dagger} \mathbf{w}(k)$.
Since $\alpha(k) \in \mathrm{Ker} \left( {S^{(\mathrm{m}, \mathrm{nom})}}^{\top} \mathcal{O}_N \right) = 
\mathrm{Ker} \left( \mathcal{O}_N \right)$, we have $CA^{i} \alpha(k) = 0$ for all $i \in \lbrace 0, \ldots, N-1 \rbrace$.
Now, if $N \geq \nu$, the observability index of the system~\eqref{eq:LTI_sys}, then we have
$\mathrm{Rank}(\mathcal{O}_N) = \mathrm{Rank}(\mathcal{O}_{N+1})$, and since 
$\mathrm{Ker}(\mathcal{O}_{N+1}) \subseteq \mathrm{Ker}(\mathcal{O}_{N})$, it follows
from the Rank-Nullity Theorem that $\mathrm{Ker}(\mathcal{O}_N) = \mathrm{Ker}(\mathcal{O}_{N+1})$.
We therefore get that $CA^{N} \alpha(k) = 0$, which implies that $\mathcal{O}_N A \alpha(k) = 0$, or in other words, 
$A \alpha(k) \in \mathrm{Ker} \left( \mathcal{O}_N \right) = \mathrm{Ker} \left( {S^{(\mathrm{m}, \mathrm{nom})}}^{\top} \mathcal{O}_N \right)$.
This is equivalent to stating that the $\mathrm{Ker} \left( {S^{(\mathrm{m}, \mathrm{nom})}}^{\top} \mathcal{O}_N \right)$
is an invariant subspace of $A$ when $N \geq \nu$, the observability index of~\eqref{eq:LTI_sys}.
Therefore, we get that $\mathbf{w}(k+1) = \left( {S^{(\mathrm{m}, \mathrm{nom})}}^{\top} \mathcal{O}_N \right) A 
{\left( {S^{(\mathrm{m}, \mathrm{nom})}}^{\top} \mathcal{O}_N \right)}^{\dagger} \mathbf{w}(k)$.

\section{Proof of Theorem~\ref{thm:well-posed_data-driven_feature_dyn}}
\label{app:proof_thm_well-posed_data-driven_feature_dyn}
We first recall that
$Y_{N,T} = \mathcal{O}_N  \left[ \begin{matrix}	x(0) & Ax(0) & \ldots & A^{T-N} x(0) \end{matrix} \right]
			   = \left[ \begin{matrix}	\mathcal{O}_N x(0) & \mathcal{O}_N A x(0) & \ldots & \mathcal{O}_N A^{T-N} x(0) \end{matrix} \right]$
%
%
We have $\mathbf{w}(k) = S^{\top} \mathcal{O}_N x(k)$ for all $k \in \mathbb{N}$. Let
$x(k) = \alpha(k) + \beta(k)$, where $\alpha(k) \in \mathrm{Ker} \left( S^{\top} \mathcal{O}_N \right)$
and $\beta(k) \in \mathrm{Ker}^{\perp} \left( S^{\top} \mathcal{O}_N \right)$.
We therefore get that $\mathbf{w}(k) = S^{\top} \mathcal{O}_N \beta(k)$.
Since $S^{\top} \mathcal{O}_N$ is a bijection from $\mathrm{Ker}^{\perp} \left( S^{\top} \mathcal{O}_N \right)$
to $\mathrm{Col} \left( S^{\top} \mathcal{O}_N \right)$, we have that $\beta(k) = 
\left( S^{\top} \mathcal{O}_N \right)^{\dagger} \mathbf{w}(k)$.

Now, since $N \geq \nu$ and $T-N \geq \mu -1$, 
we get from Theorem~\ref{thm:minimum_hor_length}
that $Y_{N,T}$ is of maximum rank and
we have $\mathbf{w}(k+1) = S^{\top} \mathcal{O}_N x(k+1) = S^{\top} \mathcal{O}_N A x(k)
= S^{\top} \mathcal{O}_N A \left( \alpha(k) + \beta(k) \right)
= S^{\top} \mathcal{O}_N A \alpha(k) + S^{\top} \mathcal{O}_N A 
\left( S^{\top} \mathcal{O}_N \right)^{\dagger} \mathbf{w}(k)$.

Suppose $\alpha(k) \in \mathrm{Ker} \left( \mathcal{O}_N \right)$, and since $N \geq \nu$
we get that $\mathrm{Ker} \left( \mathcal{O}_N \right)$ is $A$-invariant, and we get that
$A \alpha \in \mathrm{Ker} \left( \mathcal{O}_N \right)$, or in other words,
$S^{\top} \mathcal{O}_N A \alpha(k) = 0$.
Now suppose, on the other hand, that $\alpha(k) \in \mathrm{Ker}^{\perp} \left( \mathcal{O}_N \right) \cap 
\mathrm{Ker} \left( S^{\top} \mathcal{O}_N \right)$, we would then indeed have
$x(k) \in \mathrm{Ker}^{\perp} \left( \mathcal{O}_N \right)$.
Moreover, let $V \subseteq \mathrm{Ker}^{\perp} \left( \mathcal{O}_N \right)$ 
be the smallest $A$-invariant subset of $\mathrm{Ker}^{\perp}
\left( \mathcal{O}_N \right)$ containing $x(0)$. This implies that
$x(k) \in V \subseteq \mathrm{Ker}^{\perp} \left( \mathcal{O}_N \right)$.
 Since the columns of $S$
form a basis of $\mathcal{O}_N V$, we get that $\mathcal{O}_N x(k)$
can be expressed as a linear combination of the columns of $S$, i.e.,
that $\mathcal{O}_N x(k) = S v(k)$, with $v(k) \neq 0$ and we get
$x(k) = \mathcal{O}_N^{\dagger} S v(k)$. 
It then follows that $\mathcal{O}_N x(k)$ does not have a component
orthogonal to all the column vectors of $S$, which implies that
$\alpha(k) = 0$ and we get a contradiction. Therefore,
$\alpha(k) \notin \mathrm{Ker}^{\perp} \left( \mathcal{O}_N \right) \cap 
\mathrm{Ker} \left( S^{\top} \mathcal{O}_N \right)$.
Thus, we get that $\mathbf{w}(k+1) = \left( S^{\top} \mathcal{O}_N \right) A 
\left( S^{\top} \mathcal{O}_N \right)^{\dagger} \mathbf{w}(k)$.
and we can write $\vec{W} = \left( S^{\top} \mathcal{O}_N \right) A 
\left( S^{\top} \mathcal{O}_N \right)^{\dagger} W$.
It is clear from the above that $\mathrm{Col}\left( \vec{W}^{\top} \right) 
\subseteq \mathrm{Col}\left( W^{\top} \right)$,
from which we infer that $\left( S^{\top} \mathcal{O}_N \right) A 
\left( S^{\top} \mathcal{O}_N \right)^{\dagger} = \vec{W} W^{\dagger}$ 
is the global minimizer for the minimization 
problem~\eqref{eq:feature_space_dyn_minimization}
and that the minimum value is zero.
\end{appendices}
\end{document}